\documentclass[11pt,reqno]{amsart}
\usepackage{amsmath,amssymb,amsthm}
\usepackage{bbold}
\usepackage{caption}
\usepackage[usenames]{color}
\usepackage{comment}
\usepackage[shortlabels]{enumitem}
  \setlist[enumerate,1]{leftmargin=25pt}
  \setlist[itemize,1]{leftmargin=12pt}
  \setlist[description,1]{leftmargin=15pt}
\usepackage{float}\restylefloat{figure}
\usepackage{graphicx}
\usepackage{hyperref}
\usepackage{tikz}
  
\usepackage{qcircuit}
  \newcommand{\ew}[1][-1]{\ar @{} [0,#1]}
  \newcommand{\emeasure}[1]{*+[F-:<.9em>]{#1} \ew}
  
\usepackage[leftmargin=16pt,vskip=6pt]{quoting}
\usepackage{relsize}
  
  \newcommand\ms{\mathsmaller}
\usepackage{setspace}
\usepackage{url}
\newtheorem{theorem}{Theorem}[section]

\newtheorem{corollary}[theorem]{Corollary}

\newtheorem{lemma}[theorem]{Lemma}
\newtheorem{proposition}[theorem]{Proposition}

\theoremstyle{definition}

\newtheorem{convention}[theorem]{Convention}
\newtheorem{constraint}[theorem]{Constraint}
\newtheorem{definition}[theorem]{Definition}

\newtheorem{proviso}[theorem]{Proviso}
\newtheorem{remark}[theorem]{Remark}
\newtheorem{requirement}[theorem]{Requirement}

\newenvironment{lquote}
  {\list{}{\leftmargin=1.5em\rightmargin=1em}\item[]}%
  {\endlist}

\newcommand\B{\ensuremath{\mathcal B}}
\newcommand\bigox{\ensuremath{\bigotimes}}
\newcommand\bra[1]{\ensuremath{\langle#1|}}
\newcommand\braket[2]{\ensuremath{\langle#1\,|\,#2\rangle}}
\newcommand\C{\ensuremath{\mathcal C}}
\newcommand\Co{\ensuremath{\mathbb C}}
\newcommand\D{\ensuremath{\mathcal D}}
\newcommand\DO{\ensuremath{\mathrm{DO}}}
\newcommand{\dg}{^\dag}

\newcommand\Entries{\ensuremath{\textrm{Entries}}}
\newcommand\Exits{\ensuremath{\textrm{Exits}}}

\newcommand\Gsigma{\ensuremath{\overset{\sigma}{G}}}
\renewcommand\H{\ensuremath{\mathcal H}}

\newcommand\Id{\ensuremath{\mathbb 1}}
\newcommand\Inputs{\ensuremath{\textrm{Inputs}}}
\newcommand\iset[1]{\ensuremath{\big\langle#1\big\rangle}}
\newcommand\K{\ensuremath{\mathcal K}}
\newcommand\ket[1]{\ensuremath{|#1\rangle}}
\newcommand\ketbra[2]{\ensuremath{|#1\rangle\langle#2|}}

\newcommand\M{\ensuremath{\mathfrak M}}
\newcommand\Outputs{\ensuremath{\textrm{Outputs}}}
\newcommand\ox{\ensuremath{\otimes}}

\renewcommand\phi{\varphi}
\newcommand\qef{\hfill$\triangleleft$} 
\newcommand\qefhere{\tag*{$\triangleleft$}} 
\renewcommand\r{\mathop{\restriction}}
\newcommand\set[1]{\ensuremath{\{#1\}}}

\newcommand\T{\ensuremath{\mathrm{Tracks}}}
\newcommand\Tr{\ensuremath{\mathrm{Tr}}}

\newcommand\x{\ensuremath{\times}}

\title[]
      {Quantum circuits with classical channels\\
      and the principle of deferred measurements}
\author[]{Yuri Gurevich}
\address{Computer Science and Engineering\\
University of Michigan\\
Ann Arbor, MI  48109, U.S.A}
\email{gurevich@umich.edu}
\thanks{Partially supported by the US Army Research Office under W911NF-20-1-0297}
\author[]{Andreas Blass}
\address{Mathematics Department\\
University of Michigan\\
Ann Arbor, MI 48109, U.S.A.}
\email{ablass@umich.edu}

\begin{document}

\begin{abstract}
We define syntax and semantics of quantum circuits, allowing measurement gates and classical channels.
We define circuit-based quantum algorithms and prove that, semantically, any such algorithm is equivalent to a single measurement that depends only on the underlying quantum circuit.
Finally, we use our formalization of quantum circuits to state precisely and prove the principle of deferred measurements.
\end{abstract}

\maketitle
\thispagestyle{empty}

\section{Introduction} 
\label{sec:intro}

Quantum circuits play a central role in quantum computing. ``In this book,'' states the most popular textbook in the area, ``the term `quantum computer' is synonymous with the quantum circuit model of computation'' \cite[\S4.6]{NC}.

But what are quantum circuits exactly?
According to Wikipedia's Quantum Circuit page, ``A quantum circuit is a model for quantum computation in which a computation is a sequence of quantum gates, which are reversible transformations \dots'' \cite{WikiQC}.
Wikipedia may not be authoritative, but it is popular, and its reversibility claim echoes similar claims in the professional literature.
``Any $w$-qubit quantum circuit,'' according to \cite[p.~146]{De Vos}, ``is represented by a $2^w\x 2^w$ unitary matrix,'' and thus is reversible.
``Since quantum circuits are reversible, \dots'' is unreservedly stated in \cite[p.~X]{Al-Rabadi}.

Indeed, typical quantum gates are reversible.
But there are also measurement gates; see Figure~\ref{fig:cnot} for example.
In fact, measurement gates play an ever bigger role in quantum computing.

\begin{figure}[H]
\hspace*{15pt}
\Qcircuit @C=1em @R=.45em {
  &\ew &\ew &\ew &\ew &\ew &\ew
  &\emeasure{\textit{\small q=1}}\cwx[1] \\
\lstick{\ket{c}}
  &\qw &\multimeasureD{1}{\textit{\small p:=PM}} &\qw
  &\qw &\qw &\qw &\gate{Z} \\
\lstick{\ket0}
  &\gate{H} &\ghost{\textit{\small p:=PM}} &\gate{H}
  &\multimeasureD{1}{\textit{\small q:=PM}} &\gate{H}
  &\measureD{\textit{\small r:=SM}} \\
\lstick{\ket{t}}
  &\qw &\qw &\gate{H}
  &\ghost{\textit{\small q:=PM}} &\gate{H} &\qw &\gate{X} &\gate{Z}\\
  &\ew &\ew &\ew &\ew &\ew &\ew
  &\emeasure{\textit{\footnotesize p$\oplus$r=1}}\cwx[-1]
  &\emeasure{\textit{\footnotesize q=p$\oplus$r=1}}\cwx[-1]
}
\caption{\small A circuit for computing Controlled-NOT
(a slight modification of a figure from \cite{ZBL}).
PM is the qubit-parity measurement, and SM is the measurement in the standard basis.
Implicit classical channels connect each measurement with the equations where it is used.
The ancilla and the garbage to be discarded are shown on the middle line.}\label{fig:cnot}
\end{figure}
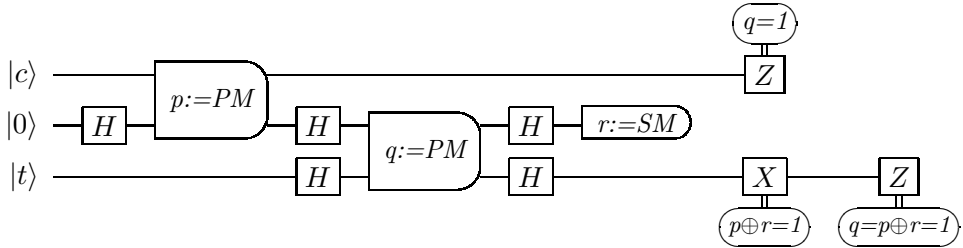

\begin{description}
\item[Question 1] What are quantum circuits exactly?
What are circuit-based quantum algorithms exactly?
\end{description}

To motivate another question, recall that a Boolean circuit with $m$ inputs and $n$ outputs computes a function of type $\set{0,1}^m \to \set{0,1}^n$, and that a general quantum circuit may have measurement gates and classical channels.
\begin{description}
\item[Question 2] What does a general quantum circuit compute?
\end{description}

Our favorite textbook on quantum computing is \cite{NC}, and in this paper we will use it as our main reference on quantum computing.
An attentive reader of \cite{NC} gets a good understanding of quantum circuits.
But even that textbook does not answer the two questions.

There is some formalization work to be done.
To begin with, it is helpful to separate syntax (circuit diagrams) from semantics (in Hilbert spaces).
The current lack of separation resembles to us the situation in classical logic before Tarski's definition of truth \cite{Tarski}.
As logicians working in quantum computing, we felt that it is our duty to develop precise definitions and analyze them.
That is what we do in this paper, though we simplified our task a little by adopting a common restriction to qubit-based circuits.


But what are such precise definitions good for?
First of all, they are useful to learners of quantum computing.
We know this from our own experience as we were such learners only a few years ago, and later one of us taught quantum computing to computer engineers.

Second, precise definitions facilitate proving general properties of quantum circuits. We prove for example that, semantically, every circuit-based quantum algorithm is equivalent to a single measurement that depends only on the underlying circuit.

Third, formalization compels careful examination of foundational issues.
One such issue is the Principle of Deferred Measurements (PDM) according to which every quantum circuit can be transformed so that no unitary gate has a measurement gate as a prerequisite.

The PDM is widely used to justify restricting attention to measurement-deferred circuits.
For the circuits free of classical channels the PDM is proved in paper \cite{AKN} which pioneered the whole issue.
We don't know any proof of a more general version of the PDM in the literature.
In fact, we don't know of a precise formulation of a more general version of the PDM in the literature.

It pains us to criticize our favorite textbook \cite{NC} on quantum computing, but the formulation of the PDM there is rather poetic
\cite[\S4.4]{NC}.
In various forms, that formulation is restated in the literature and used as if it were a proper theorem, e.g. \cite{BW,GC,JST,Tao+5,WikiDM}.

In \S\ref{sec:dm}, we formulate precisely and prove the principle of deferred measurements.

Finally, we need to mention a large body of sophisticated work on space-bounded quantum computations \cite{FR,GR,MW,PJ,Watrous} employing specialized quantum Turing machines and addressing the issues like eliminating intermediate measurements using pseudorandom number generators. This literature is primarily on structural computational complexity. As far as we can judge, it does not offer practical algorithms for the general PDM problem.

\section{Preliminaries}\label{sec:prelim}

By default, in this paper, Hilbert spaces are finite dimensional.
We take density operators on Hilbert space \H\ to be nonzero, positive semidefinite, Hermitian, linear operators on \H.
A density operator $\rho$ is \emph{normalized} if its trace is 1. We use possibly non-normalized density operators to represent (mixed) states in \H. A general density operator $\rho$ represents the same state as its normalized version $\rho/\Tr(\rho)$.

As in \cite[\S2.2.3]{NC}, a (quantum) \emph{measurement} $M$ on a Hilbert space $\H$ is an indexed family \iset{A_i: i\in I} of linear operators on \H\ where $\sum_{i\in I} A_i^\dag A_i$ is the identity operator $\Id_\H$ on \H.
The index set $I$ is the set of possible (classical) outcomes of $M$. (If there is only one possible outcome, then the unique operator is necessarily unitary.)

More generally a \emph{measurement} $M$ from a Hilbert state \H\ to a Hilbert space $\H'$ is an indexed family \iset{A_i: i\in I} of linear transformations $A_i: \H\to \H'$ such that $\sum_{i\in I} A_i^\dag A_i = \Id_{\H}$.
If the measurement $M$ is performed in state $\rho$,
the probability of the outcome $i$ is $\Tr(A_i\rho A_i\dg) / \Tr(\rho)$.
And if the outcome is $i$, then the post-measurement state is (represented by) $A_i\rho A_i\dg$ in $\H'$.

\begin{convention}\label{cnv:renorm}
The density operator $A_i \rho A_i\dg$ may not have trace 1, even if $\rho$ does, but this density operator is most convenient for our purposes. While we allow any positive scalar multiple of $A_i \rho A_i\dg$ to represent the same state,
usually we will represent the post-measurement state as $A_i\rho A_i\dg$ where $\rho$ represents the pre-measurement state. \qef
\end{convention}

\begin{convention}\label{cnv:counts}
Let $\H,\K$ be Hilbert spaces.
A linear operator $A$ on \H\ counts also as a linear operator on $\H\ox\K$, being tacitly identified with $A\ox\Id_\K$. Accordingly, a measurement $M$ on \H\ with index set $I$ and linear operators $A_i$ on \H\ counts as a measurement on $\H\ox\K$ with the same index set $I$ and with linear operators $A\ox\Id_\K$ on $\H\ox\K$. \qef
\end{convention}


\begin{remark}
An analogous convention can be formulated for operators $A:\H\to\H'$ and measurements $M$ from \H\ to $\H'$.
The notation, however, is getting more complicated.
In many cases and certainly in the cases of interest to us in this paper, the Hilbert spaces \H\ and $\H'$ are isomorphic and can be identified along an appropriate isomorphism.
This avoids inessential detail and simplifies exposition.
\end{remark}

\section{Syntax} 
\label{sec:syntax}

We describe syntactic circuits underlying quantum circuits that work with qubits. That is, every input node produces a single qubit, every output node consumes a single qubit, and every producer (an input node or gate exit) transmits to its consumer (a gate entry or output node) a single qubit.

\begin{definition}\label{def:syncir}
A \emph{syntactic circuit} consists of the following components.
\begin{enumerate}
\item Disjoint finite sets of \emph{input nodes}, \emph{output nodes}, and \emph{gates}.
\item For each gate $G$, two disjoint nonempty finite sets of the same cardinality, the set $\Entries(G)$ of the \emph{entries} of $G$ and the set $\Exits(G)$ of the \emph{exits} of $G$.\\
    The sets associated with any gate are disjoint from those associated with any other gate and from the sets in (1).
    The input nodes and gate exits will be called \emph{producers}. The gate entries and output nodes will be called \emph{consumers}.
\item A one-to-one function Bind from the set of producers onto the set of consumers. If an exit $x$ of gate $G_1$ is bound to an entry $y$ of gate $G_2$ (i.e., Bind$(x)=y$), we say that $G_1$ is a \emph{quantum source} for $G_2$ and write $G_1\prec_q G_2$.
\item A binary relation on the gates called the \emph{classical source relation} and denoted $G_1\prec_c G_2$.
\end{enumerate}
It is required that the following \emph{source} relation on gates
\[ G_1 \prec G_2\quad\text{if}\quad
  G_1\prec_q G_2\ \lor\ G_1 \prec_c G_2 \]
be acyclic. \qef
\end{definition}

A gate $G_1$ is a \emph{prerequisite} for gate $G_2$ if $G_1 \prec^* G_2$ where $\prec^*$ is the transitive closure of the source relation $\prec$.
Since $\prec$ is acyclic, so is $\prec^*$.

View each relationship $G_1 \prec_c G_2$ as a \emph{channel} from gate $G_1$ to gate $G_2$.
If $G$ has at least one incoming channel, then $G$ is a \emph{classically controlled} gate, in short a CC gate; otherwise $G$ is non-CC gate.

We presumed that Bind is defined on all producers. This situation is similar to that with Boolean circuits: unbound producers may be made bound by providing additional output nodes.

But, in a sense, syntactic circuits underlying Boolean circuits are more general \cite{G242,G244}. The definition above reflects two special aspects of quantum circuits which go beyond syntax.
\begin{itemize}
\item In Boolean circuits, information may flow from one producer to multiple consumers. In quantum circuits, by the no-cloning theorem of quantum theory \cite[\S12.1.1]{NC}, a producer's (quantum) output cannot generally be duplicated to supply multiple consumers. Hence the requirement that Bind be a function.
    (As in Boolean circuits, every consumer needs exactly one producer, so Bind is a bijection.)
\item A Boolean gate often has more entries than exits; think of a conjunction gate for example. It may also have more exits than entries. A quantum gate, without loss of generality (see \S\ref{sec:q} in this connection), transforms one state of a quantum system to another state of the same quantum system. Since we work with qubits-to-qubits gates, the number of entry qubits is equal to the number of exit qubits.
\end{itemize}

Consider a syntactic circuit $\C$.

\begin{definition}\label{def:stage}
A \emph{stage} (suggesting a stage of a computation) of \C\ is a set $S$ of gates closed under prerequisites, so that $F\prec^* G\in S$ implies $F\in S$.
A gate $G$ is \emph{S-ready} (suggesting that it is ready to fire at stage $S$) if all its prerequisites are in $S$ but $G\not\in S$.
The \emph{exits} of $S$ are the producers $x$ such that $x$ is an input node or an exit of an $S$-gate and the consumer $\mathrm{Bind}(x)$ is an output node or an entry of a gate outside of $S$. \qef
\end{definition}

\begin{lemma}\label{lem:stage}
Let $n$ be the number of input nodes of \C. Then any stage $S$ of \C\ has exactly $n$ exits, and \C\ has exactly $n$ output nodes.
\end{lemma}

\begin{proof}
The second claim follows from the first when $S$ contains all the gates.
We prove the first claim by induction on the cardinality $|S|$ of $S$. The case $|S|=0$ is obvious; the exits of $S$ are exactly the input nodes.
Suppose that $|S|>0$. By the acyclicity of the prerequisite relation, there is a gate $G\in S$ which isn't a prerequisite of any $S$ gate. By the induction hypothesis, the stage $V = S - \{G\}$ has exactly $n$ exits. Let $r$ be the arity of $G$.
When we add $G$ to $V$, the set of exits loses the $r$ producers for the entries in $G$ and gains the $r$ exits of $G$. Hence $S$ has exactly $n$ exits.
\end{proof}

\begin{corollary}\label{cor:truncate}
Any stage $S$ of a syntactic circuit gives rise to a syntactic circuit in its own right, a truncated version of the original circuit, with the gates of $S$, with the original input nodes and with the consumers of the exits of $S$ playing the role of output nodes.
\end{corollary}

\section{Semantics}
\label{sec:q}

The computation model of Boolean circuits is well known. Every Boolean circuit computes a function of type $\set{0,1}^m \to \set{0,1}^n$.
In this section, we define the computation model of quantum circuits.
In particular, we answer the question: What exactly does a quantum circuit compute?

If \H\ is a Hilbert space and $I$ a nonempty set, then $\H^{\ox I}$ is the tensor product $\bigox_{i\in I} \H_i$ where each $\H_i$ is $\H$.

The following convention would allow us a more uniform view of quantum circuit gates.

\begin{convention}\label{cnv:umeas}
A unitary operator $U$ is identified with a measurement with a single outcome whose only operator is $U$.
\end{convention}

\begin{definition}\label{def:qcrct}
A \emph{quantum circuit} \C\ is a syntactic circuit $\C_0$ together with the following assignments.
\begin{enumerate}
\item Each non-CC gate $G$ is assigned a single measurement $M(G)$, called the \emph{$G$-measurement}, consisting of linear transformations from $(\Co^2)^{\ox\Entries(G)}$ to $(\Co^2)^{\ox\Exits(G)}$.
        An outcome of $M(G)$ is a \emph{$G$-outcome}, and $O(G)$ is the set of $G$-outcomes.
\item Each CC gate is assigned a finite set of measurements, called \emph{$G$-measurements}, with disjoint index sets, consisting of linear transformations from $(\Co^2)^{\ox\Entries(G)}$ to $(\Co^2)^{\ox\Exits(G)}$.
    An outcome of any of the $G$-measurements is a \emph{$G$-outcome}, and $O(G)$ is the set of $G$-outcomes.
    In addition, $G$ is assigned a \emph{selection} function $\sigma_G$ that, given outcomes \iset{f(G'): G'\prec_c G} of all classical sources of $G$, picks a $G$-measurement.
\end{enumerate}
A gate $G$ is \emph{unitary} if every $G$-measurement is unitary. \qef
\end{definition}

The disjointness requirement in clause~(2) is a convenience that simplifies notation. As a result, any $G$-outcome determines the $G$-measurement producing the outcome.

The selection function determines, at runtime, the $G$-measurement to be executed in accordance with information from the classical sources of $G$. Without loss of generality, we assume that every classical source $G'\prec_c G$ sends to $G$ the actual $G'$-outcome $f(G')$. In applications, only some $h(f(G'))$ may be sent. For example, if outcomes $f(G')$ are natural numbers, then only the parity of $f(G')$ might be sent. But that auxiliary function $h$ may be built into the selection function of $G$.

\begin{remark}\label{rem:umeas}
Convention~\ref{cnv:umeas} is convenient in the present section. It allows a more uniform treatment of gates. But it is just a convention, and it is not really necessary. It may be dropped, so that we do distinguish between a unitary transformation $U$ and the measurement with a single outcome whose only transformation is $U$.
Then a gate $G$ may be assigned one or several measurements, in which case it is a \emph{measurement gate}. Alternatively, it may be assigned one or several unitary transformations, in which case it is a \emph{unitary gate}. A unitary gate $G$ has no $G$-outcomes. If it is classically controlled, then its selection function $\sigma_G$ picks one of the assigned unitary transformations depending on the outcomes of the classical sources of $G$. As you read the rest of this section, you'll see that the other necessary changes are rather obvious. \qef
\end{remark}

\begin{proviso}\label{prv:local}
Quantum circuits are local physical systems devoid of long-distance communication. \qef
\end{proviso}

Notice that Hilbert spaces $(\Co^2)^{\ox\Entries(G)}$ and $(\Co^2)^{\ox\Exits(G)}$ are in general different.
The definition of a measurement on a Hilbert space naturally generalizes to that of a measurement from one Hilbert space to another; see \S\ref{sec:prelim}.
But working systematically with the more general definition would make our exposition notationally awkward. What can we do?

Well, $(\Co^2)^{\ox\Entries(G)}$ and $(\Co^2)^{\ox\Exits(G)}$ have the same dimension and therefore are isomorphic.
There is in general no canonical isomorphism from $(\Co^2)^{\ox\Entries(G)}$ to $(\Co^2)^{\ox\Exits(G)}$.
But every bijection from $\Entries(G)$ to $\Exits(G)$ engenders an isomorphism from $(\Co^2)^{\ox\Entries(G)}$ to $(\Co^2)^{\ox\Exits(G)}$.
Choosing such a bijection for every gate and linearly ordering the inputs of \C\ give rise to so-called registers, also known as wires or timelines, typically drawn horizontally in pictures like  Fig~\ref{fig:cnot}.
Such  registers are common%
\footnote{We experimented with formalizing quantum circuits with timelines in \cite{G242} and without timelines in \cite{G244}.}
in the literature. They are not necessary in this section but we will use them in \S\ref{sec:dm}.
For now, we just need to choose some isomorphisms, not necessarily induced by bijections. Accordingly, we adopt the following proviso.

\begin{proviso}\label{prv:choose}
To simplify our presentation, we will presume that a quantum circuit comes with an isomorphism $\eta_G: (\Co^2)^{\ox\Entries(G)} \to (\Co^2)^{\ox\Exits(G)}$ for every gate $G$.
Furthermore, for each gate $G$, we will identify the Hilbert spaces $(\Co^2)^{\ox\Entries(G)}$ and $(\Co^2)^{\ox\Exits(G)}$ along the isomorphism $\eta_G$. \qef
\end{proviso}

In the rest of this section, let \C\ be a quantum circuit and $\H = (\Co^2)^{\ox \Inputs(\C)}$. We extend the selector function $\sigma$ to non-CC gates $G$ in the only possible way: Given any outcomes \iset{f(G'): G' \prec_c G} of the classical sources of $G$ (namely, none, as $G$ has no classical sources $G'$), $\sigma_G$ picks the unique $G$-measurement $M(G)$.

\begin{definition}\label{def:track}
A \emph{track} for \C\ is a function $f$ that assigns to each gate $G$ a $G$-outcome $f(G)$ subject to the following \emph{coherence} requirement:
\[ \text{$f(G)$ is an outcome of the measurement}\quad \sigma_G\iset{f(G'): G' \prec_c G}. \qefhere \]
\end{definition}
As far as classically controlled unitary gates are concerned, the coherence means that the control works as intended.

A \emph{stage} of \C\ is a stage of the underlying syntactic circuit. Recall that a gate $G$ is ready at stage $S$ if all its prerequisites are in $S$ but $G\notin S$.

\begin{lemma}\label{lem:choose}
There is a unique way to associate with every stage $S$ of \C\ an isomorphism $\eta_S: \H \to (\Co^2)^{\ox \Exits(S)}$ such that $\eta_\emptyset = \Id_\H$ and,
if $G$ is $S$-ready, then $\eta_{S\cup\{G\}} = \eta_G \circ \eta_S$.
\end{lemma}

The composition $\eta_G \circ \eta_S$ makes sense because of Convention~\ref{cnv:counts}.

\begin{proof}
Induction on the cardinality of $S$.
\end{proof}


The lemma implies that the gates of the circuit \C\ work on tensor factors of \H; by Convention~\ref{cnv:counts}, they work on \H. Note that $(\Co^2)^{\ox \Outputs(\C)}$ is identified with \H\ via the isomorphism $\eta_{\C}$ of Lemma~\ref{lem:choose}.

We begin to address the question: What does our quantum circuit \C\ compute?

A stage $S$ of a quantum circuit \C\ determines and represents a possible stage in a computation of the circuit, where $S$ comprises the gates that have already fired.
A gate $G$ can be fired at stage $S$ if and only if $G$ is $S$-ready, so that $S$ contains all prerequisites of $G$ but not $G$ itself.
Notice that the $S$-ready gates form an antichain (in the $\prec^*$   ordering) which means that none of them is a prerequisite for another.
But if a stage $S$ is reached in a computation of \C, then perhaps only some of the $S$-ready gates fire.

\begin{definition}\label{def:qbout}
A \emph{gate bout} is a nonempty set of gates which form an antichain. \qef
\end{definition}

Think of a gate bout $B$ as a generalized gate. The entries of $B$ are the entries of all $B$-gates, and the exits of $B$ are the exits of all $B$-gates.

\begin{definition}
Let $B$ be a gate bout.
\begin{enumerate}
\item Every tensor product $\bigox_{G\in B} M_G$, where $M_G$ is a $G$-measurement, is a \emph{$B$-measurement}, and the outcomes of $\bigox_{G\in B} M_G$  are \emph{outcomes} of $B$.
    Thus the set of outcomes of $B$ is $O(B)= \prod_{G\in B} O(G)$.
    A gate $G'$ is a \emph{classical source} for $B$, symbolically, $G'\prec_c B$, if $G'$ is a classical source for some $B$-gate.
\item The \emph{selection} function $\sigma_B$ is the function that, given outcomes $f(G')$ of all classical sources $G'$ of $B$, picks the $B$-measurement
    \[ \bigox_{G\in B} \sigma_G\iset{f(G'): G'\prec_c G}. \qefhere \]
\end{enumerate}
\end{definition}

Any computation of a quantum circuit works in sequential time, step after step.
At each stage $S$ of the computation, a bout of $S$-ready gates is fired.
But which bout? Some decisions have to be made.
At each stage $S$, we decide which of the $S$-ready gates fire, and nature decides what measurement results will be produced.
To reflect our decisions, we introduce the following notion.

\begin{definition}\label{def:qalg}
A \emph{schedule} of a quantum circuit is a sequence
\[ X = (X_1; X_2; X_3; \dots; X_T) \]
  of gate bouts such that every gate set
\[ X_{\le t} =  \bigcup \{X_s: s\le t\} \]
is a stage, and $X_{\le T}$ contains all the gates.
(Notice that all the gates in $X_{t+1}$ are $(X_{\le t})$-ready.)
A quantum circuit \C\ equipped with a fixed schedule is a \emph{circuit algorithm} $\C_X$.
\qef
\end{definition}
\noindent
The intent is that the bouts $X_1, \dots, X_T$ are to be fired in that order.

We stipulate that an input for any circuit algorithm $\C_X$ is a possibly-mixed state in $\H = (\Co^2)^{\ox\Inputs(\C)}$ given by a density operator in $\DO(\H)$.

Given an input $\rho$ and schedule $X$, any computation of $\C_X$ on $\rho$ fires every gate and, in that sense, \emph{realizes} some track. And every track is realized in at most one
computation of $\C_X$.

\begin{definition}\label{def:qcum}
For any schedule $X = (X_1; X_2; X_3; \dots; X_T)$ for \C\ and any track $f\in\T(\C)$, the \emph{cumulative operator} is defined as
\[
C_X^f = A_T \circ A_{T-1} \circ \cdots \circ A_2 \circ A_1
\]
where each $A_t$ is the operator associated with the outcome $f\r X_t = \iset{f(G): G\in X_t}$ in the measurement $\sigma_{X_t}\iset{f(G'): G'\prec_c X_t}$.
Finally,
\[
\M(\C_X) = \iset{C_X^f: f\in \T(\C)}
\]
is the \emph{aggregate measurement} of $\C_X$. \qef
\end{definition}
It is easy to check that $\M(\C_X)$ is indeed a measurement.

\begin{theorem}[Reduction]\label{thm:qred}
Executing a circuit algorithm $\C_X$  on input  $\rho\in\DO(\H)$ and performing the aggregate measurement $\M(\C_X)$ in state $\rho$ have exactly the same effect.
More explicitly, for every track $f$ and input $\rho$ for \C, we have the following.
\begin{enumerate}
\item The probability that a computation of $\C_X$ realizes track $f$ is equal to the probability of outcome $f$ in the measurement $\M(\C_X)$.
\item If a computation of $\C_X$ realizes track $f$, then the resulting final state is $C_X^f\rho (C_X^f)\dg$.
\end{enumerate}
\end{theorem}

\begin{proof}
Let $X = (X_1; X_2; X_3; \dots; X_T)$ and $A_1, \dots, A_T$ be as in Definition~\ref{def:qcum}.
The probability, according to quantum mechanics, that a computation of $\C_X$ on input $\rho$ realizes track $f$ is
\[
\frac{\Tr(A_1\rho A_1\dg)}{\Tr(\rho)} \cdot
  \frac{\Tr(A_2A_1\rho A_1\dg A_2\dg)}{\Tr(A_1\rho A_1\dg)}\cdot
  \frac{\Tr(A_3A_2A_1\rho A_1\dg A_2\dg A_3\dg)}
       {\Tr(A_2A_1\rho A_1\dg A_2\dg)} \cdots
= \frac{\Tr(C_X^f\rho (C_X^f)\dg)}{\Tr(\rho)}
\]
which is the probability of outcome $f$ in the measurement $\M(\C_X)$.

Suppose that a computation of $\C_X$ realizes $f$.
The computation successively applies $A_1$ to $\rho$, $A_2$ to $A_1\rho A_1\dg$, $A_3$ to $A_2A_1\rho A_1\dg A_2\dg$, \dots. The final state is $C_X^f \rho (C_X^f)\dg$.
\end{proof}

Our goal in the rest of this section is to show that the behavior of $\C_X$, as summarized in its aggregate measurement, depends only on the circuit \C, not on the schedule $X$, i.e. not on our choices of which ready gates to fire first.
To this end, call schedules $X,Y$ of \C\ \emph{equivalent} if $\C_X$ and $\C_Y$ have the same aggregate measurement.

\begin{lemma}\label{lem:qeq}
Let $X$ be a schedule $(X_1; \dots; X_T)$ of \C, and suppose that a bout $X_t$ is the disjoint union $B_1\sqcup B_2$ of bouts $B_1,B_2$, so that
\[ X = (X_1; \dots; X_{t-1}; B_1\sqcup B_2; X_{t+1}; \dots; X_T). \]
Then the schedule
\[ Y = (X_1; \dots; X_{t-1}; B_1; B_2; X_{t+1}; \dots; X_T)\]
is equivalent to $X$.
\end{lemma}

\begin{proof}
We need to prove that $\C_X^f = \C_Y^f$ for every track $f\in \T(\C)$. So let $f$ be an arbitrary track for \C. Let
\[ \C_X^f = (A_T\circ\cdots\circ A_{t+1})\circ A_t \circ
            (A_{t-1}\circ\cdots\circ A_1) \]
as in Definition~\ref{def:qcum}. Further, let $B_0 = X_t = B_1\sqcup B_2$ and $L_0 = A_t$.
Then $L_0$ is the operator associated with the outcome $f\r B_0$ in the measurement $\sigma_{B_0}\iset{f(G'): G'\prec_c B_0}$.

It suffices to prove that $L_0 = L_2\circ L_1$ where, for $j\in \set{0,1,2}$, $L_j$ is the operator associated with the outcome $f\r B_j$ in the measurement $\sigma_{B_j}\iset{f(G'): G'\prec_c B_j}$.
Let $E_j = \Entries(B_j)$ and $\H_j = (\Co^2)^{E_j}$.
The equality $B_0 = B_1\sqcup B_2$ implies $E_0 = E_1 \sqcup E_2$ and therefore $\H_0 = \H_1 \ox \H_2$. Accordingly $L_0 = L_1\ox L_2$.

It remains to show that $L_2\circ L_1 = L_1\ox L_2$. We have
\begin{align*}
& (L_2\circ L_1)(x_1 \ox x_2)  = L_2(L_1(x_1\ox x_2))
 = L_2(L_1(x_1)\ox x_2)\\
& = L_1(x_1) \ox L_2(x_2) = (L_1 \ox L_2)(x_1 \ox x_2)
 = L_0(x_1\ox x_2). \qedhere
\end{align*}
\end{proof}

\begin{theorem}[Equivalence]\label{thm:qeq}
Every two schedules $X,Y$ over the same quantum circuit \C\ are equivalent.
\end{theorem}

\begin{proof}
Let $n$ be the number of gates in \C.
Call a schedule $Y$ \emph{linear} if every bout of $Y$ contains a single gate, so that $Y$ can be identified with the sequence of gates $G_1, G_2, \dots, G_n$ in the order they appear in $Y$. Observe that a sequence of gates $G_1, G_2, \dots, G_n$ is a schedule if and only if it is \emph{coherent} in the sense that it respects the prerequisite relation: if $G_i\prec^* G_j$ then $i<j$.

First prove that every schedule $X = (X_1; X_2; \dots; X_T)$ is equivalent to a linear schedule. Induct on $n-T$. If $n-T = 0$, $X$ is already linear. Otherwise, split some non-singleton bout $X_t$ into the disjoint union of bouts $B_1$ and $B_2$, and then use Lemma~\ref{lem:qeq} and the induction hypothesis.

Second, prove that every two linear schedules are equivalent.
By Theorem~\ref{thm:comb} and the observation above, any linear schedule can be transformed to any other linear schedule by adjacent transpositions with all intermediate sequences being legitimate schedules.
Accordingly it remains to prove that linear schedules $X$ and $Y$ are equivalent if $Y$ is obtained from $X$ by one adjacent transposition.

To this end, let
\begin{align*}
X &= G_1, \dots, G_{t-1},\ G_t, G_{t+1},\ G_{t+2}, \dots, G_n\\
Y &= G_1, \dots, G_{t-1},\ G_{t+1}, G_t,\ G_{t+2}, \dots, G_n
\end{align*}
Neither $G_t$ nor $G_{t+1}$ is a prerequisite for the other. Indeed, if $G_t\prec^*G_{t+1}$ then $Y$ would be incoherent, and if $G_{t+1}\prec^*G_t$ then $X$ would be incoherent. Thus the set \set{G_t, G_{t+1}} is a bout. Let
\[Z = G_1, \dots, G_{t-1}, \set{G_t, G_{t+1}}, G_{t+2}, \dots, G_n. \]

By Lemma~\ref{lem:qeq}, $X$ is equivalent to $Z$, and $Y$ is equivalent to $Z$. Hence $X,Y$ are equivalent.
\end{proof}

Theorem~\ref{thm:qeq} justifies the following definition.

\begin{definition}
The \emph{aggregate measurement} $\M(\C)$ of a quantum circuit \C\ is the aggregate measurement of (any of) the \C-based circuit algorithms. \C\ \emph{computes} $\M(\C)$. \qef
\end{definition}

\begin{remark}
The aggregate measurement of a quantum circuit \C\ provides the semantics of \C.
The number of linear operators of the aggregate measurement may be exponential in the size of \C, but recall that the semantics of a Boolean circuit \B\ is given by a Boolean function, a truth table, which may be exponentially large in the number of \B\ gates. \qef
\end{remark}

\section{Deferring measurements}
\label{sec:dm}

In the book \cite{NC}, Nielsen and Chuang put forward the Principle of Deferred Measurement (PDM):
\begin{quoting}
``Measurements can always be moved from an intermediate stage of a quantum circuit to the end of the circuit;
if the measurement results are used at any stage of the circuit then the classically controlled operations can be replaced by conditional quantum operations" \cite[\S4.4]{NC}.
\end{quoting}
They obviously distinguish between measurements and conditional quantum operations, which must be unitary.
To be on the same page with them, we drop Convention~\ref{cnv:umeas} here.
See Remark~\ref{rem:umeas} in this connection.

To illustrate the PDM, Nielsen and Chuang transform the teleportation circuit
\begin{figure}[H]
\hspace{10pt}
\Qcircuit @C=2em @R=.2em {
\lstick{\ket\psi\quad }
&\ctrl{1} &\gate{H} &\measuretab{M} &\cw &\cctrl{2} \\
&\targ &\qw &\measuretab{N} & \cctrl{1} \\
&\qw &\qw &\qw &\gate{X^N}
  &\gate{Z^M} &\qw &\rstick{\ket\psi}
  \inputgroupv{2}{3}{.8em}{.8em}{\ket{\beta_{00}}\quad } }
\caption*{\small Figure 1.13 in \cite{NC}}
\end{figure}
\noindent
to the circuit
\begin{figure}[H]
\hspace{10pt}
\Qcircuit @C=2em @R=.2em {
\lstick{\ket\psi\quad }
&\ctrl{1} &\gate{H} &\qw &\ctrl{2} &\measuretab{M}\\
&\targ &\qw &\ctrl{1} &\qw &\measuretab{N} \\
&\qw &\qw &\gate{X^N} &\gate{Z^M} &\qw &\rstick{\ket\psi}
  \inputgroupv{2}{3}{.8em}{.8em}{\ket{\beta_{00}}\quad } }\\
\caption*{\small Figure 4.15 in \cite{NC}}
\end{figure}
\noindent
But, contrary to the first circuit, the second doesn't teleport a quantum state over a distance. Specifically, in the second circuit, \ket\psi\ is ``teleported'' only as far as the size of the controlled-$X$ and controlled-$Z$ gates; these gates must extend from Alice to Bob.

\noindent
\begin{quoting}
``Of course, some of the interpretation of this circuit as performing ‘teleportation’ is lost, because no classical information is transmitted from Alice to Bob, but it is clear that the overall action of the two quantum circuits is the same, which is the key point" \cite[\S4.4]{NC}.
\end{quoting}
The fact is that distant teleportation is impossible without classical transfer of information \cite[\S2.4.3]{NC}.
Under Proviso~\ref{prv:local}, quantum circuits are local physical systems devoid of long-distance communication. In the rest of this section, the proviso remains in force.

For simplicity, in this section, we work with registers (a.k.a.\ wires or timelines) described in \S\ref{sec:q}.

As indicated in the introduction, formulations of the PDM, similar to the one above, are found in the literature,
but we have not seen any formulation there which is more precise or explicit.

Let \C\ denote a given circuit. We are seeking a measurement-deferred version \D\ of \C, and we express ``deferring'' as follows.

\begin{requirement}[Deferral requirement]\label{req:d}
The deferred version \D\ satisfies

No unitary gate has a measurement gate as a prerequisite. \qef
\end{requirement}
\noindent
And it is desirable of course that the construction of \D\ from \C\ is feasible.

In the absence of classical channels, the PDM is established in \cite{AKN}.
The following precise but naive form of the PDM follows from our results in \S\ref{sec:q}.

\begin{proposition}\label{prp:DM}
Any quantum circuit \C\ can be transformed to an equivalent quantum circuit \D\ on the \C\ qubits such that \D\ satisfies the deferral requirement.
\end{proposition}

\begin{proof}
If \C\ satisfies the deferral requirement, set $\D = \C$.
Otherwise, let \D\ be the single-gate circuit that computes the aggregate measurement of \C.
In a trivial way, \D\ satisfies the deferral requirement.
By Theorem~\ref{thm:qred}, \C\ and \D\ are equivalent.
\end{proof}

While our formalization above arguably fits the informal PDM,
it is presumably not what Nielsen and Chuang (and other authors) intended.
But what did they intend?
This is not an easy question to answer, but let us make a couple of points expressing how we view the intent in question.
First, it seems that  Nielsen and Chuang restrict attention to quantum circuits satisfying the following constraint.

\begin{constraint}\label{cns:channel}
Every classical channel goes from a measurement gate to a unitary gate; there are no classically controlled measurement gates.
\end{constraint}

By Definition~\ref{def:track}, a track of a quantum circuit is an assignment of a $G$-outcome to every measurement gate $G$. Definition~\ref{def:track} imposes a coherence requirement on tracks, but Constraint~\ref{cns:channel} implies that all assignments are coherent.
In the rest of this section, by default, quantum circuits satisfy the constraint.

Second, the measurement deferral procedure is expected just to defer the measurements of \C\ but otherwise keep the structure of \C\ intact to the extent possible; see \cite[Exercise~4.35]{NC} in this connection. The following definition captures one aspect of that expectation.

\begin{definition}[Commensurate]\label{def:com}
Circuits \C\ and \D\ are \emph{commensurate} if there is a one-to-one correspondence $\zeta$ between the measurements of \C\ and those of \D\ such that, for every measurement $M$ in \C, the measurement $\zeta(M)$ has the outcomes of $M$ (and possibly some extra outcomes). Such a $\zeta$ is a \emph{commensuration correspondence}.
\end{definition}
\noindent
A commensuration correspondence $\zeta$ allows us to view a track $f$ of \C\ as a track of \D: $f(\zeta(M)) = f(M)$.

A question arises in what sense \C\ and \D\ are equivalent. The equivalence notion of \S\ref{sec:q} is too strong. Indeed, it takes an ancilla to defer the measurement in a circuit like
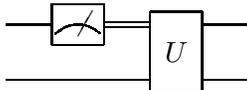
\[
\Qcircuit @C=1.6em @R=.75em {
&\meter & \cghost{U} & \qw \\
& \qw &\multigate{-1}{U} & \qw }
\]
(See the proof of Lemma~\ref{lem:dm} for how an ancilla is used.)
Accordingly, \C\ and \D\ may have different aggregate measurements. The following definition gives the most natural relation for the purpose of the PDM.

\begin{definition}\label{def:fs}
A circuit \C\ is \emph{faithfully simulated} by a circuit \D, symbolically $\C\propto\D$, if the following conditions hold.
\begin{enumerate}
\item  \D\ works with the qubits of \C, the \emph{principal qubits} of \D, and may employ additional qubits, \emph{ancillas}, initially in state \ket0.
\item \C\ and \D\ are commensurate under some commensuration correspondence $\zeta$ (which identifies the tracks $f$ of \C\ with some tracks $\zeta(f)$ of \D).
\item For every pure input \ket\psi\ and every track $f$ for \C, circuits \C\ and \D\ realize $f$ with the same probability.
\item For every pure input \ket\psi\ and every track $f$ for \C,  the computations of \C\ and \D\ determined by $f$ compute same output when the ancillas of \D\ are traced out. \qef
\end{enumerate}
\end{definition}

\begin{corollary}
Suppose that \D\ faithfully simulates \C. Then every track $g$ of \D\ which does not have the form $\zeta(f)$ for any track $f$ of \C\ is of probability zero for any input state of \D.
\end{corollary}

\begin{proof}
By item~(3) of Definition~\ref{def:fs}, the probabilities of the tracks $\zeta(f)$ add up to 1.
\end{proof}

Recall that a measurement over a Hilbert space \H\ is projective if it consists of mutually orthogonal projection operators. We call a projective measurement $P = \iset{P_i: i\in I}$ \emph{complete} if every $P_i$ projects \H\ to a one-dimensional space $\H_i$, and we call  $P$ \emph{standard} if each $\H_i$ is spanned by a single computational basis vector.

\begin{lemma}\label{lem:dm}
Every quantum circuit \C\ with only standard measurements is faithfully simulated by a quantum circuit \D\ such that \D\ satisfies the deferral requirement.
\end{lemma}

\begin{proof}
We proceed by induction on the number of (possibly classically controlled) unitary gates having a measurement gate among their prerequisites; call such unitary gates red and call other unitary gates green. If there are no red gates, we are done. Otherwise, (using acyclicity) let $G$ be one of the red gates having no red gates among its prerequisites. To complete the induction, it suffices to show that $G$ can be replaced with green gates by deferring its measurement prerequisites.

Without loss of generality, \C\ has no green gates. Indeed, if there are green gates in \C, schedule them before all measurement gates.
Let $\C'$ be the rest of \C\ (after the green gates). If $\D'$ faithfully simulates $\C'$ then by prefixing $\D'$ with the green gates, we get a circuit \D\ that faithfully simulates \C.

Without loss of generality, every qubit is measured separately in \C.
Indeed, if $M$ is a standard measurement involving registers $R_1,\dots, R_k$, replace $M$ with standard measurements $M_1, \dots, M_k$ on registers $R_1, \dots, R_k$.
Also replace any channel from $M$ to a unitary gate $G$ by channels from $M_1, \dots, M_k$ to $G$.
Adjust the selection function of $G$ so that $G$ works the same with $M_1, \dots, M_k$ as it did with $M$.

As far as classical outcomes are concerned, every standard one-qubit measurement can be viewed as a binary variable.
Further, without loss of generality, different prerequisite measurements of gate $G$ are on different registers; there is no point in measuring the same qubit twice in the same computational basis without any intervening unitary.
(More formally, the second measurement produces the same classical outcome and the same post-measurement state as the first, so it can be deleted, and any classical channel from it to a later gate can be treated as a channel from the first measurement to the same later gate and be handled by the same selection function.)

We illustrate the remainder of the proof on the example where gate $G$ is preceded by four measurements $p,q,r$ and $s$; see Figure~\ref{fig:dm}.
There are channels to $G$ from measurements $q,r$ but not from $p,s$, and $G$ shares registers with $r,s$ but not with $p,q$.
In general, a red gate could have several prerequisites like $p$, several like $q$, several like $r$, and several like $s$. (``Several'' includes the possibility of zero.) All of these can be handled the same way as the $p,q,r$ and $s$ in the example.

In accordance with Remark~\ref{rem:umeas}, $G$ comes with unitary operators $U_{\sigma(q,r)}$ where $q,r\in\{0,1\}$.
Let $\Gsigma$ be the unitary gate computing the transformation $\ket{jklx} \mapsto \ket{j}\ox U_{\sigma(j,k)}\ket{klx}$.
We shall show that the left side \C\ of Figure~\ref{fig:dm} is faithfully simulated by the right side \D\ with deferred measurements.
\begin{figure}[H]
\begin{minipage}{0.4\textwidth}
\begin{align*}
\Qcircuit @C=1em @R=.75em {
  &\ew &\ew &\emeasure{\ms{\sigma(q,r)}}\cwx[3] \\
  &\measureD{\ms p} &\qw\hspace{18pt}\cdots&\hspace{30pt}\cdots \\
  &\measureD{\ms q}&\qw\hspace{18pt}\cdots&\hspace{30pt}\cdots \\
  &\measureD{\ms r}  &\qw &\multigate{3}{G} &\qw \\
  &\measureD{\ms s}  &\qw &\ghost{G} &\qw\\
  &\cdots\cdots           &\hspace{6pt}\cdots
    &\nghost{G} &\cdots \\
  &\qw &\qw               &\ghost{G} &\qw }
\end{align*}
\end{minipage}
\begin{minipage}{0.1\textwidth}
\begin{center}
\vspace{10pt} $\propto$
\end{center}
\end{minipage}
\begin{minipage}{0.4\textwidth}
\vspace{40pt}
\begin{align*}
\Qcircuit @C=1em @R=.75em {
  &\qw &\qw &\qw &\measureD{\ms p} \\
  &\qw &\qw &\multigate{4}{\Gsigma} &\measureD{\ms q} \\
  &\ctrl{4} &\qw &\ghost{\Gsigma} &\qw \\
  &\qw &\ctrl{4} &\ghost{\Gsigma} &\qw \\
  &\cdots\cdots\quad &\nghost{\Gsigma} &\hspace{40pt}\cdots \\
  &\qw &\qw &\ghost{\Gsigma} &\qw \\
\lstick{\ket0} &\targ &\qw &\qw &\measureD{\ms r} \\
\lstick{\ket0} &\qw &\targ &\qw &\measureD{\ms s}
}
\end{align*}
\end{minipage}
\caption{\small Deferring measurements past one unitary gate}\label{fig:dm}
\end{figure}
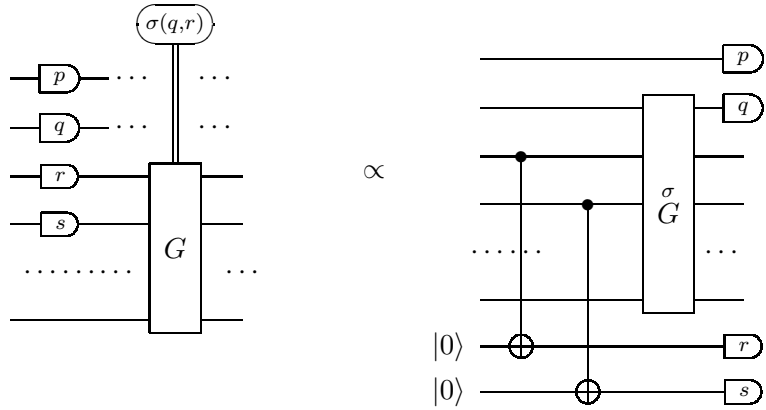

\smallskip
Indeed, consider any pure initial state for \C, say $\sum_{ijkl} a_{ijkl}\ket{ij} \ox \ket{kl x_{ijkl}}$.
In \C, the state evolution is
\[\sum_{ijkl} a_{ijkl}\ket{ij} \ox \ket{kl x_{ijkl}} \to
  \ket{pq}\ox\ket{rs x_{pqrs}} \to
   \ket{pq}\ox U_{\sigma(q,r)}\ket{rs x_{pqrs}}.
\]
In \D, the state evolution is
\begin{align*}
&\sum_{ijkl} a_{ijkl}\ket{ij}\ox\ket{kl x_{ijkl}}\ox\ket{00}
 \to \sum_{ijkl} a_{ijkl}\ket{ij}\ox\ket{kl x_{ijkl}}\ox\ket{kl}\to\\
&\sum_{ijkl} a_{ijkl}\ket{ij}\ox
  \left(U_{\sigma(j,k)}\ket{kl x_{ijkl}}\right) \ox\ket{kl} \to
\ket{pq}\ox\left(U_{\sigma(q,r)}\ket{rsx_{pqrs}}\right)\ox\ket{rs}.
\end{align*}
It follows that \D\ faithfully simulates \C.
\end{proof}

\begin{remark}
It is common that, as in Figure~\ref{fig:cnot} and in the teleportation examples, every classically-controlled unitary gate $G$ comes with a single unitary $U$ and the selector function of $G$ has only two possible values: $U$ and the identity.
This restriction on the selector functions is consistent with the PDM discussion in \cite{NC},
but we do not impose it. \qef
\end{remark}

\begin{theorem}\label{thm:dm}
Every quantum circuit \C\ is faithfully simulated by a quantum circuit \D\ such that \D\ satisfies the deferral requirement.
\end{theorem}

\begin{proof}
Induction on the number of nonstandard measurements in \C. If the number is zero, use Lemma~\ref{lem:dm}.
Otherwise let $M = \iset{A_i: i\in I}$ be a nonstandard measurement in \C.

Without loss of generality, we may assume that $I$ is an initial segment \set{0,1, \dots, |I|-1} of natural numbers.
Let $l = \lceil \log_2|I|\rceil$.
Expand \C\ with ancilla qubits $\alpha_0, \dots, \alpha_{l-1}$ initially in state \ket0.
For each natural number $k<l$, let $\H_k$ be the state space of $\alpha_k$.
The $2^l$ vectors \ket{a_0\dots a_{l-1}}, where each $a_k$ is 0 or 1, form the computational basis of $\bigox_{k=0}^{l-1} \H_k$.
For each natural number $i<2^l$ let \ket{i} be the vector \ket{a_0\dots a_{l-1}} such that $a_0\dots a_{l-1}$ is the binary representation of $i$.

Let $\H = (\Co^2)^{\ox Inputs(\C)}$ and define a transformation
$U\ket\psi = \sum_{i<|I|} (A_i\ket\psi \ox \ket{i})$ from\ \H\ to $\H\ox \bigox_i \H_i$ which is easily seen to be unitary.
For each $i<2^l$, let $P_i$ be the projection $\ketbra ii$, and let $P$ be the standard projective measurement \iset{P_i: i<2^l} on $\bigox_i \H_i$. Recall that, according to Convention~\ref{cnv:counts}, $P$ counts as a measurement on $\H \ox \bigox_i \H_i$.

Furthermore, consider the following two events: outcome $i\in I$ of $M$ in state \ket\psi\ and outcome $i$ of $P$ in state $U\ket\psi$.
It is easy to see that the two events have the same probability $p_i$, and if the post-measurement state of $M$ is $\frac{A_i\ket\psi}{\sqrt{p_i}}$, then the post-measurement state of $P$ is $\frac{A_i\ket\psi}{\sqrt{p_i}} \ox \ket{i}$.

Let $\C'$ be the circuit obtained from the expanded $\C$ by replacing $M$ with $U$ followed by $P$; every channel from $M$ to a unitary gate $G$ becomes a channel from $P$ to $G$.
Then $\C$ is commensurate with $\C'$; the commensuration correspondence $\zeta$ sends $M$ to $P$ and is the identity elsewhere.

$\C'$ faithfully simulates \C\ and has fewer nonstandard gates.
By the induction hypothesis, there is a circuit \D\ satisfying the deferral requirement and faithfully simulating $\C'$. It also faithfully simulates \C.
\end{proof}

\appendix
\section{Combinatorics}
\label{sec:A}

For the reader's convenience, we reproduce here Appendix~A of \cite{G242}.

Call a linear order $<$ on a poset (partially ordered set) $\mathcal S = (S,\prec)$ \emph{coherent} if $a<b$ whenever $a\prec b$.

A linear order $<$ on a finite set $S$ can be transformed into any other linear order $<'$ on $S$ by adjacent transpositions. In other words, there is a sequence $<_1$, $<_2, \dots, <_k$ of linear orders such that $<_1$ is $<$, and $<_k$ is $<'$, and every $<_{i+1}$ is obtained from $<_i$ by transposing one pair of adjacent elements of $<_i$.
The question arises whether, if $<$ and $<'$ are coherent with a partial order $\prec$, the intermediate orders $<_i$ in the transposition sequence can also be taken to be coherent with $\prec$. The following theorem answers this question affirmatively.

\begin{theorem}\label{thm:comb}
Any coherent linear order on a finite poset $\mathcal S = (S,\prec)$ can be transformed into any other coherent linear order on $\mathcal S$ by adjacent transpositions with all intermediate orders being coherent.
\end{theorem}

\begin{proof}
Fix a finite poset $(S,\prec)$. We start with an observation that if two elements $u,v$ are ordered differently by two coherent linear orders then $u,v$ are incomparable by $\prec$. Indeed, if $u,v$ were comparable then one of the two linear orders would not be coherent.

Define the distance $D(<,<')$ between two coherent linear orders $<$ and $<'$ to be the number of $(<,<')$ differentiating pairs $u,v$ such that $u<v$ but $v<'u$. We claim that if $D(<,<')\ge1$ then there is a  $(<,<')$ differentiating pair $u,v$ such that $u,v$ are adjacent in ordering $<$. It suffices to prove that if $u,v$ is a $(<,<')$ differentiating pair and $u<w<v$ then either $u,w$ or $w,v$ is a  $(<,<')$ differentiating pair, so that $w<'u$ or $v<'w$. But this is obvious. If $u<'w<'v$ then $u<'v$ which is false.

We prove the theorem by induction on the distance $D(<,<')$ between two given coherent linear orders $<$ and $<'$.
If $D(<,<')=0$, the two orders are identical and there is nothing to prove. Suppose $D(<,<')=d\ge1$.

By the claim above there exist $u<v$ such that $u,v$ are adjacent in $<$ but $v<'u$.
By the observation above, $u,v$ are incomparable by $\prec$. Let $<''$ be the order obtained from $<$ by transposing the adjacent elements $u$ and $v$. $<''$ is coherent because $u,v$ is the only $(<,<'')$ differentiating pair and because $u,v$ are incomparable by $\prec$.

It remains to prove that $<''$ can be transformed into $<'$  by adjacent transpositions with all intermediate linear orders respecting $\prec$. But this follows from the induction hypothesis. Indeed, $D(<'',<')=d-1$ because the $(<'',<')$ differentiating pairs are the same as the $(<,<')$ differentiating pairs, except for $u,v$.
\end{proof}

\end{document}